\documentclass[%
reprint,
bibnotes,
aps,
prl,
]{revtex4-2}
\usepackage{amsmath,amssymb,mathtools}
\usepackage{xcolor}
\definecolor{mplblue}{HTML}{1F77B4}
\definecolor{mplgreen}{HTML}{2CA02C}
\definecolor{mplred}{HTML}{D62728}
\definecolor{mplpurple}{HTML}{9467BD}
\definecolor{mplolive}{HTML}{BCBD22}
\usepackage[colorlinks=true, allcolors=mplblue]{hyperref}
\usepackage{graphicx}
\usepackage{dcolumn}
\usepackage{bm}
\usepackage{dsfont}
\usepackage{hologo}
\usepackage{physics}
\usepackage{url}
\usepackage{mathrsfs}
\usepackage{amsthm}
\usepackage{booktabs}
\usepackage{multirow}
\usepackage{adjustbox}
\usepackage{algorithm}
\usepackage{algpseudocode}

\graphicspath{{resources/img/}}

\newtheorem{theorem}{Theorem}

\renewcommand\norm[1]{\left\Vert#1\right\Vert}
\renewcommand\abs[1]{\left\vert#1\right\vert}
\renewcommand\var[1]{\mathrm{var}#1}
\renewcommand\trace{\operatorname{Tr}}
\renewcommand{\d}{\mathrm{d}}

\begin{document}

\preprint{APS/123-QED}

\title{Efficient characterization of general
	Gottesman-Kitaev-Preskill qubits}

\author{Vojt\v{e}ch Kucha\v{r}}
\author{Petr Marek}%
\affiliation{%
	Department of Optics, Palack\'y University, 17. listopadu 1192/12, 779 00
	Olomouc, Czech Republic
}%

\date{\today}

\begin{abstract}
	Practical utilization of Gottesman-Kitaev-Preskill (GKP) qubits requires not only the preparation of logical basis states, but also the ability to prepare and evaluate arbitrary logical qubit superpositions. This is typically done via quantum state tomography, which is resource-intensive. We introduce a family of positive semidefinite Hermitian operators, one for each point on the logical Bloch sphere, whose unique zero-eigenvalue ground states are the corresponding ideal GKP qubit states. We show that the expectation value of each operator serves as a witness of non-Gaussianity, and corresponds to twice the logical infidelity for states in the ideal logical GKP subspace. Furthermore, the truncated finite-dimensional counterparts of these operators yield physical approximations of arbitrary logical GKP states as their ground states. Crucially, the evaluation of the proposed operators requires only three sets of homodyne measurements, making this framework practical for both experimental characterization and numerical optimization of state preparation circuits.
\end{abstract}

\maketitle


Bosonic systems, such as modes of traveling light, are considered as candidates for the realization of quantum computation~\cite{nielsenQuantumComputationQuantum2012, asavanantOpticalQuantumComputers2022}. Their advantage lies in their processing speed and the ability to interact: with the basic tools of linear optics it is possible to generate entangled states of practically unlimited size~\cite{yoshikawaInvitedArticleGeneration2016,asavanantGenerationTimedomainmultiplexedTwodimensional2019, larsenFaulttolerantContinuousvariableMeasurementbased2021b, aghaeeradScalingNetworkingModular2025}. The remaining open challenges lie in the preparation and manipulation of the logical qubits~\cite{bourassaBlueprintScalablePhotonic2021, hastrupProtocolGeneratingOptical2022, endoNongaussianQuantumState2023, konnoLogicalStatesFaulttolerant2024, larsenIntegratedPhotonicSource2025, endoHighRateFourPhoton2025}. There are several possible approaches to encoding information in such systems; the Gottesman-Kitaev-Preskill (GKP) code is favored by many due to its conceptually straightforward approach to perform error correction and implement elementary Clifford gates, both by application of feasible Gaussian operations~\cite{gottesmanEncodingQubitOscillator2001, schmidtQuantumErrorCorrection2022, hastrupAnalysisLossCorrection2023, conradGottesmankitaevpreskillCodesLattice2022}.

The ideal logical GKP qubit states are defined as nonphysical infinite superpositions of quadrature eigenstates. Their finite energy approximations can be found either as superpositions of squeezed states~\cite{mensenPhasespaceMethodsRepresenting2021}, or as the ground state of a specific operator restricted to a particular Fock space dimension~\cite{marekGroundStateNature2024}, and they can be conditionally prepared in a quantum boson sampling architecture with iterative mixing steps~\cite{vasconcelosAllopticalGenerationStates2010, weigandGeneratingGridStates2018, suConversionGaussianStates2019, winnelDeterministicPreparationOptical2024}. Evaluating the quality of GKP qubit states, be it for the purpose of analyzing experimental results or numerically optimizing the state preparation circuits, can be done with the help of their logical stabilizers or their GKP nonlinear squeezing~\cite{brauerGeneralizedSqueezingWitness2025, hanamuraStellarRankControl2025a}. Such approaches are needed, because quantum state fidelity is ill-defined as a measure of similarity of two states when the target state is nonphysical~\cite{kucharNonlinearSqueezingSuperpositions2025}. Additionally, they are also practical from the point of evaluation of experimental data, as they only require two sets of measurements, one for each quadrature operator.

However, preparation and evaluation of the logical basis states are not sufficient for practical applications. Specific superpositions of logical basis states are also required both for the encoding of logical information and as resources for measurement-induced realization of non-Clifford gates~\cite{bartlettEfficientClassicalSimulation2002a, bravyiUniversalQuantumComputation2005, baragiolaAllgaussianUniversalityFault2019, konnoNoncliffordGateOptical2021, hastrupUnsuitabilityCubicPhase2021}. While such arbitrary qubit states can be defined using only the basis states, their characterization currently requires evaluating their fidelity to some approximate state. Furthermore, any such evaluation requires costly quantum state tomography~\cite{lvovskyContinuousvariableOpticalQuantumstate2009}.

In this letter we introduce an efficient way to characterize and evaluate qubits encoded in the GKP basis. The approach is based on the nonlinear squeezing paradigm, which builds upon the framework of Gaussian squeezing and applies it to various non-Gaussian states, such as cubic phase states, Schr\"odinger cat states, or Fock states~\cite{kalaCubicNonlinearSqueezing2022, brauerGenerationQuantumStates2021, brauerCatabilityMetricEvaluating2025, provaznikWitnessesNonGaussianFeatures2026}. We generalize the nonlinear squeezing for GKP states and, for each possible logical GKP state, present a positive semidefinite Hermitian operator with the corresponding ideal logical GKP state as its unique ground state. We show that these operators distinguish between different logical GKP states, can be used to find physical approximations, and are suitable for direct measurement and evaluation.
\par

A quantum two-level system is completely described by the Pauli operators \(\boldsymbol{\hat{\sigma}}=(\hat{\sigma}_x,\hat{\sigma}_y,\hat{\sigma}_z)\) satisfying \(\left[ \hat{\sigma}_i, \hat{\sigma}_j \right] = 2i \sum_k \varepsilon_{ijk} \hat{\sigma}_k\) and \(\hat{\sigma}_i^2 = \hat{1}\) for \(i,j,k \in \{x,y,z\}\). The set of pure states is then uniquely mapped to the Bloch sphere
\begin{equation}
	\hat{\rho} = \frac{1}{2}\left( \hat{1}+\mathbf{u}\cdot\boldsymbol{\hat{\sigma}} \right),
\end{equation}
where \(\mathbf{u}=(u_x,u_y,u_z)\in\mathbb{R}^3\) is the unit Bloch vector and where \(\cdot\) denotes the dot product. Equivalently, defining the operator
\begin{equation}
	\hat{O}_q = \hat{1} - \left( u_x \hat{\sigma}_x + u_y \hat{\sigma}_y + u_z \hat{\sigma}_z \right)
\end{equation}
and denoting eigenstates of \(\hat{\sigma}_z\) as \(\ket{0}\) and \(\ket{1}\), an arbitrary pure state \(\ket{\psi}=\cos\theta\ket{0}+e^{i\phi}\sin\theta\ket{1}\) is found as the lowest eigenvalue eigenstate of this operator, where
\begin{subequations}\label{eq:scparameters}
	\begin{align}
		 & \cos\theta = \sqrt{\frac{1+u_z}{2}}, \\
		 & e^{i\phi}\sin\theta =
		\frac{u_x + i u_y}{\sqrt{2(1+u_z)}},
	\end{align}
\end{subequations}
offering a more natural parameterization of the Bloch sphere with \(\theta\in[0,\pi/2]\) and \(\phi\in[0,2\pi)\).

A general pure qubit state in the GKP basis can be expressed as
\begin{equation}\label{eq:GKP qubit}
	|\psi\rangle = \cos\theta |0_\mathrm{L}\rangle + e^{i\phi}\sin\theta|1_\mathrm{L}\rangle,
\end{equation}
where the basis states
\begin{subequations}\label{eq:GKP basis}
	\begin{align}
		 & |0_\mathrm{L}\rangle \propto \sum_{s \in \mathbb{Z}} |x = 2 s \sqrt{\pi}\rangle,   \\
		 & |1_\mathrm{L}\rangle \propto \sum_{s \in \mathbb{Z}} |x = (2s+1) \sqrt{\pi}\rangle
	\end{align}
\end{subequations}
are defined as infinite superpositions of specific eigenstates of the quadrature operator $\hat{x}$, that together with its canonically conjugate operator $\hat{p}$ satisfying $[\hat{x},\hat{p}] = i$ describe the harmonic oscillator upon which the GKP qubits are defined. The logical basis states in Eqs.~\eqref{eq:GKP basis} are also simultaneous eigenstates of the stabilizer operators \(\hat{X}^2,\hat{Z}^2\) and span the logical GKP space on which the logical Pauli operators are defined as
\begin{equation}
	\hat{X} = e^{-i\hat{p}\sqrt{\pi}}, ~ \hat{Z}= e^{i \hat{x}\sqrt{\pi}}, ~ \hat{Y} = i\hat{X}\hat{Z}.
\end{equation}
Therefore, the basis states can be found as eigenstates of suitable linear combinations of these operators~\cite{marekGroundStateNature2024}. This approach can be generalized for arbitrary qubit states as follows:

\begin{theorem}\label{thm}
	An arbitrary ideal GKP qubit state, as defined in Eq.~\eqref{eq:GKP qubit}, is the unique zero-eigenvalue ground state of the positive semidefinite operator
	\begin{equation}\label{eq:general GKP qubit}
		\hat{O}_\mathrm{GKP}(\mathbf{u}) = \hat{O}_1 + \hat{1} - \mathbf{u}\cdot\hat{\mathbf{O}},
	\end{equation}
	where $\mathbf{u}\in\mathbb{R}^3$ is a unit Bloch vector with the angles $\theta,\phi$ defined via Eqs.~\eqref{eq:scparameters}, \(\hat{\mathbf{O}} = (\hat{O}_x,\hat{O}_y,\hat{O}_z)\), and the component operators are defined as
	\begin{subequations}
		\begin{align}
			\label{eq:Ooperators}
			\hat{O}_1 & = \hat{1}-\frac{1}{6}\left(\hat{X}^2 + \hat{X}^{2 \dag} + \hat{Y}^2 + \hat{Y}^{2 \dag} + \hat{Z}^2 + \hat{Z}^{2 \dag}\right),                                \\
			\hat{O}_x & = \frac{\hat{X}+\hat{X}^{\dag}}{{2}},~\hat{O}_y =\frac{\hat{Y} + \hat{Y}^{\dag}}{{2}}, ~ \hat{O}_z = \frac{\hat{Z}+\hat{Z}^{\dag}}{{2}}.\label{eq:gkppaulis}
		\end{align}
	\end{subequations}
\end{theorem}
\begin{proof}
	Defining the Hermitian generators
	\[
		\hat G_x = \sqrt{\pi}\,\hat p,\qquad
		\hat G_y = \sqrt{\pi}\,(\hat x-\hat p),\qquad
		\hat G_z = \sqrt{\pi}\,\hat x,
	\]
	any logical Pauli operator \(\hat P_k\in\{\hat X,\hat Y,\hat Z\}\), with corresponding generator \(\hat G_k\) for \(k\in\{x,y,z\}\), satisfies
	\[
		\hat P_k^2+\hat P_k^{2\dagger}
		=
		2\cos(2\smash{\hat G_k})
		=
		2\hat 1-4\sin^2(\smash{\hat G_k}).
	\]
	This implies the inequality \(\hat O_1 = \frac{2}{3}\sum_{k}\sin^2(\smash{\hat G_k})\succeq 0\), which is saturated on the ideal GKP logical subspace. The component operators read
	\[
		\hat O_k=\frac{\hat P_k+\hat P_k^\dagger}{2}=\cos(\smash{\hat{G}_k}) \implies \hat O_k^2 = \hat 1-\sin^2(\smash{\hat{G}_k}).
	\]
	On the full phase space, the displacement algebra implies \(\hat P_i\hat P_j=-\hat P_j\hat P_i\) for \(i\ne j\), and the same anticommutation relations propagate to the symmetrized component operators, \(\{\hat O_i,\hat O_j\}=0\) for \(i\ne j\). Consequently,
	\begin{align*}
		(\mathbf u\cdot\hat{\mathbf O})^2
		 & =
		\sum_{k}u_k^2\hat O_k^2 + \sum_{i\neq j} u_i u_j \hat{O}_i \hat{O}_j \\
		 & =
		\sum_{k}u_k^2\bigl(\hat 1-\sin^2(\smash{\hat G_k})\bigr)+ \underset{=\,0}{\underbrace{\sum_{i< j} u_i u_j \{\hat O_i, \hat O_j\}}}.
	\end{align*}
	Expanding \(\frac12(\hat 1-\mathbf u\cdot\hat{\mathbf O})^2 =\frac12\hat 1 -\mathbf u\cdot\hat{\mathbf O}+\frac12(\mathbf u\cdot\hat{\mathbf O})^2\) and using \(\sum_k u_k^2=1\) to cancel the constant terms gives
	\[
		\hat O_{\mathrm{GKP}}(\mathbf u) = \frac12\bigl(\hat 1-\mathbf u\cdot\hat{\mathbf O}\bigr)^2 + \sum_{k}\left(\frac23+\frac{u_k^2}{2}\right)\sin^2(\smash{\hat G_k}),
	\]
	a sum of squares with strictly positive coefficients, so \(\hat O_{\mathrm{GKP}}(\mathbf u)\succeq 0\). For arbitrary \(\ket{\psi}\), \(\hat O_\mathrm{GKP}(\mathbf u)\ket{\psi}=0\) if and only if \(\ket{\psi}\) is in the kernel of both summands. This is equivalent to \(\sin(\smash{\hat G_k})\ket\psi=0\) for all \(k\) and \((\hat 1-\mathbf u\cdot\hat{\mathbf O})\ket\psi=0\). The first condition restricts \(\ket\psi\) to the ideal GKP code space, where \(\hat O_x,\hat O_y,\hat O_z\) act as the logical Pauli operators, and the second term uniquely selects the ideal GKP qubit with Bloch vector \(\mathbf u\).
\end{proof}

For ease of reference, simplified forms of the target complement \(\hat{O}_\mathrm{GKP}-\hat{O}_1=\hat{1}-\mathbf{u}\!\cdot\!\hat{\mathbf{O}}\) for important logical states are listed in Table \ref{tab:operators}. For magic states, such as
\begin{equation}\label{eq:magic_state}
	\ket{H_\mathrm{L}} = \frac{1}{\sqrt{2}}\left(\ket{0_\mathrm{L}}+e^{i\pi/4}\ket{1_\mathrm{L}}\right),
\end{equation}
we obtain irreducible forms with two or three non-commuting cosines, which is a direct consequence of their asymmetric structure in the phase space.

\begin{table}[htbp]
	\centering
	\renewcommand{\arraystretch}{1.5} 
	\begin{tabular*}{\columnwidth}{@{\extracolsep{\fill}}c c l@{}}
		\toprule
		\textbf{Target}        & \textbf{Vector}             & \textbf{Complement} $ \hat{O}_\mathrm{GKP}-\hat{O}_1$ \\
		\midrule
		$|0_\mathrm{L}\rangle$ & $(0,0,1)$                   & $2\sin^2\left(\frac{\sqrt{\pi}}{2}\hat{x}\right)$  \\

		$|1_\mathrm{L}\rangle$ & $(0,0,-1)$                  & $2\cos^2\left(\frac{\sqrt{\pi}}{2}\hat{x}\right)$  \\

		$|+_\mathrm{L}\rangle$ & $(1,0,0)$                   & $2\sin^2\left(\frac{\sqrt{\pi}}{2}\hat{p}\right)$  \\

		$|-_\mathrm{L}\rangle$ & $(-1,0,0)$                  & $2\cos^2\left(\frac{\sqrt{\pi}}{2}\hat{p}\right)$  \\

		$|H_\mathrm{L}\rangle$ & $\frac{1}{\sqrt{2}}(1,1,0)$ & $\hat{1} - \frac{1}{\sqrt{2}}\Bigl[\cos(\sqrt{\pi}\hat{p}) + \cos\big(\sqrt{\pi}(\hat{x}-\hat{p})\big) \Bigr]$ \\

		\bottomrule
	\end{tabular*}
	\caption{Analytic forms of the target operator complements for fundamental GKP logical states alongside their corresponding Bloch sphere vectors $\mathbf{u}=(u_x, u_y, u_z)$.}
	\label{tab:operators}
\end{table}

In certain high-symmetry cases, the full form of \(\hat{O}_1\) is not necessary to use. Ref.~\cite{marekGroundStateNature2024} demonstrates that the logical state $|0_\mathrm{L}\rangle$ is the zero-eigenvalue eigenstate of $2\sin^2(\hat{x}\sqrt{\pi}/2) + 2\sin^2(\hat{p}\sqrt{\pi})$. This is a special case of Eq.~\eqref{eq:general GKP qubit}: the first term is exactly $\hat{1}-\hat{O}_z$, while the second term isolates the momentum-quadrature penalty component of $\hat{O}_1$. This reduced form is entirely sufficient because the missing penalties would only redundantly reinforce the logical requirements already enforced by $\hat{O}_z$.
\par

We have shown that Eq.~\eqref{eq:general GKP qubit} can be used to describe an arbitrary ideal qubit. To fully evaluate its effectiveness for characterization of physical states, we define its \(N\)-dimensional truncated counterpart as
\begin{equation}\label{eq:trunc}
	\hat{O}_{\mathrm{GKP}}^{[N]} = \sum_{n,m=0}^{N-1} \ketbra{n}{n} \hat{O}_{\mathrm{GKP}} \ketbra{m}{m}.
\end{equation}
This allows us to immediately identify specific finite dimensional approximations of logical GKP qubits in \(N\)-dimensional truncated Fock spaces as the ground states of \(\hat{O}_{\mathrm{GKP}}^{[N]}\) for varying \(N\). Examples for the magic state per Eq.~\eqref{eq:magic_state} are shown in Fig.~\ref{fig:magic_state}, showing gradual convergence to the expected phase space structure~\cite{hastrupAnalysisLossCorrection2023}.

\begin{figure}
	\centering
	\includegraphics[width=\columnwidth]{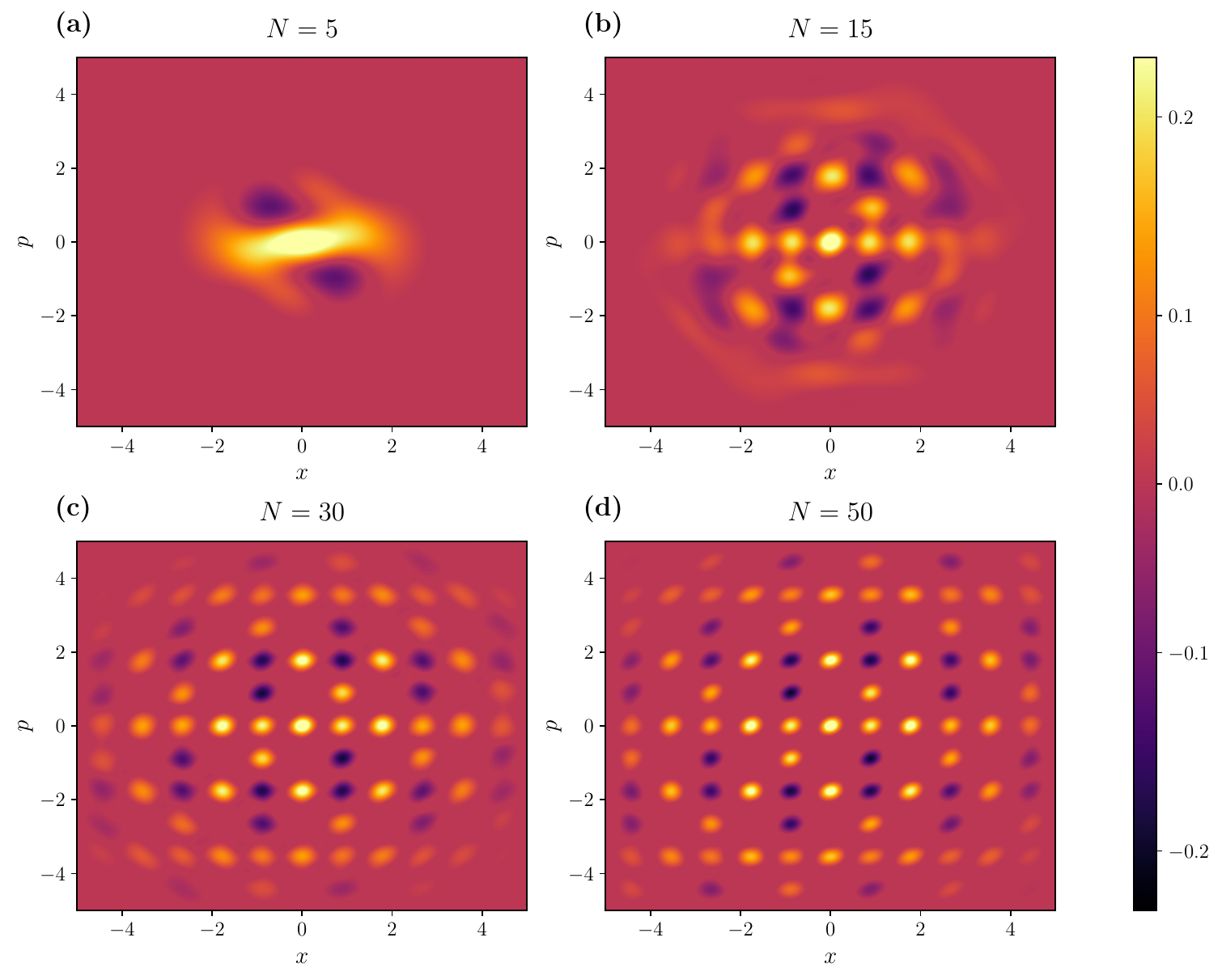}
	\caption{Wigner functions of ground states of \(\hat{O}_{\mathrm{GKP}}^{[N]}\) for \(\ket{H_\mathrm{L}}\), and \textbf{(a)} \(N = 5\), \textbf{(b)} \(N = 15\), \textbf{(c)} \(N = 30\) and \textbf{(d)} \(N = 50\). Generated using Table \ref{tab:operators} and truncated to \(N\)-dimensional Fock spaces per Eq.~\eqref{eq:trunc}.}
	\label{fig:magic_state}
\end{figure}

To verify the proposed method across the entire logical Bloch sphere, we perform a systematic numerical analysis. We evaluate \(\hat{O}_\mathrm{GKP}\) and its ground state for \(\sim 1000\) uniformly sampled target states, which are then sequentially ordered using a nearest-neighbor heuristic (see Supplemental Material~\cite{supplement} for details). By calculating the expectation value of each ordered operator evaluated across all ordered ground states, we can directly compare the result to the pairwise logical infidelity of the respective Bloch vectors, which is given by
\begin{equation}
	1-\mathcal{F}_{ij} = \frac{1}{2}\left(1-\mathbf{u}_{i}\cdot\mathbf{u}_{j}  \right),
\end{equation}
where \(\mathbf{u}_i = (u_{x,i}, u_{y,i}, u_{z,i})\) and \(\mathbf{u}_j = (u_{x,j}, u_{y,j}, u_{z,j})\) are the unit Bloch vectors of the two logical states. This comparison is shown in Fig.~\ref{fig:exp_val}.

\begin{figure}
	\centering
	\includegraphics[width=\columnwidth]{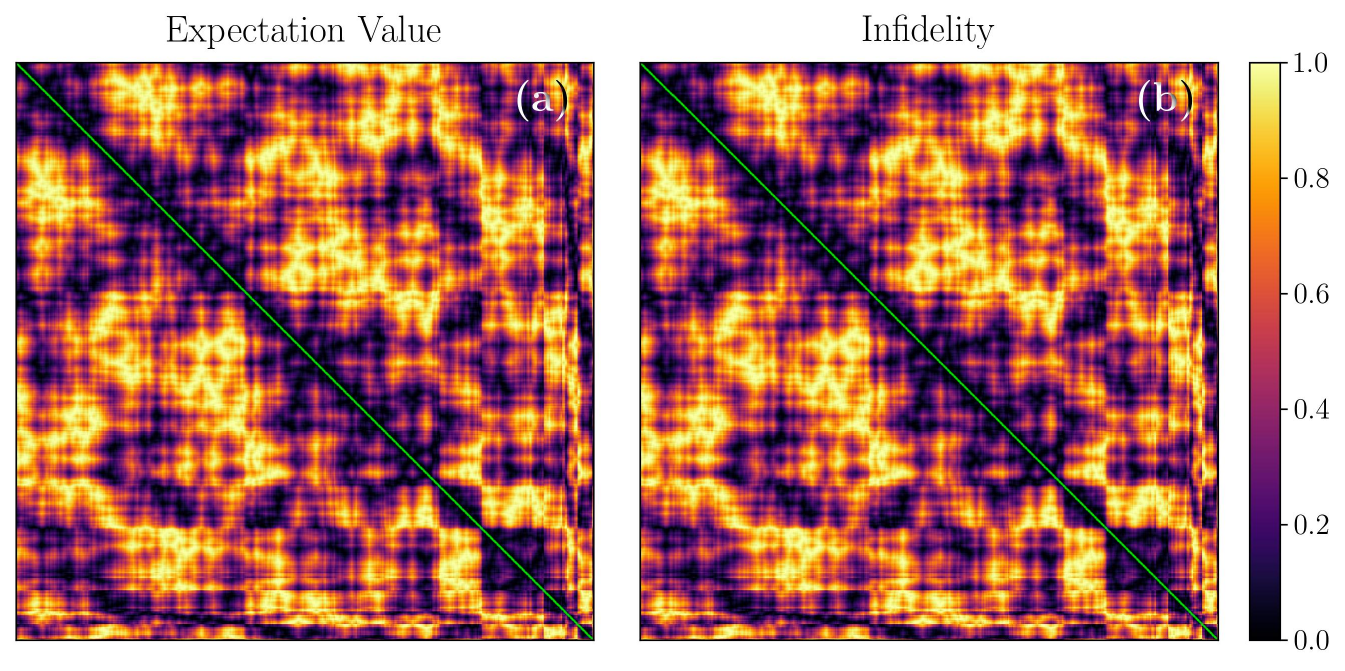}
	\caption{Comparison of expectation values and infidelities for logical GKP qubits. Indexing rows and columns of each heatmap as \(i\) and \(j\) respectively, heatmap \textbf{(a)} represents \(\langle\psi_i|\hat{O}_\mathrm{GKP}^{[N]}(\mathbf{u}_j)|\psi_i\rangle\), where \(\ket{\psi_i}\) is the ground state of \(\hat{O}_\mathrm{GKP}^{[N]}(\mathbf{u}_i)\), \(N=150\), and heatmap \textbf{(b)} represents \(1-\mathcal{F}_{ij}\), where \(\mathbf{u}_i\) and \(\mathbf{u}_j\) are unit vectors of the ordered logical states evenly sampled on the Bloch sphere. Both heatmaps are normalized to \([0,1]\) to allow for direct comparison, minima in each row and column in both heatmaps are marked by green.}
	\label{fig:exp_val}
\end{figure}

While the artifacts of ordering evenly spaced points on a sphere onto a line are clearly visible, the important takeaways are that:
\begin{itemize}
	\item the minima in each row and column of the heatmap, marked by green, fall perfectly onto the diagonal, confirming that the expectation value of the operator is minimized for the intended target state,
	\item the ordering-induced structure is identical in both plots, indicating a linear relationship between \(\langle\hat{O}_\mathrm{GKP}\rangle\) and the logical qubit infidelity.
\end{itemize}

More specifically, by viewing Fig.~\ref{fig:exp_val} as a dataset of expectation value-infidelity pairs, we find a strong linear relationship between the two quantities (see Fig.~\ref{fig:exp_val_corr}), with a numerically extrapolated asymptotic slope \(\langle\hat{O}_\mathrm{GKP}\rangle = \left( 2.00\pm 0.02 \right)(1-\mathcal{F})\) as $N\rightarrow \infty$. The convergence is supported by vanishing correlation error and diverging mutual information of the two quantities, see Supplemental Material~\cite{supplement} for further details.

\begin{figure}
	\centering
	\includegraphics[width=\columnwidth]{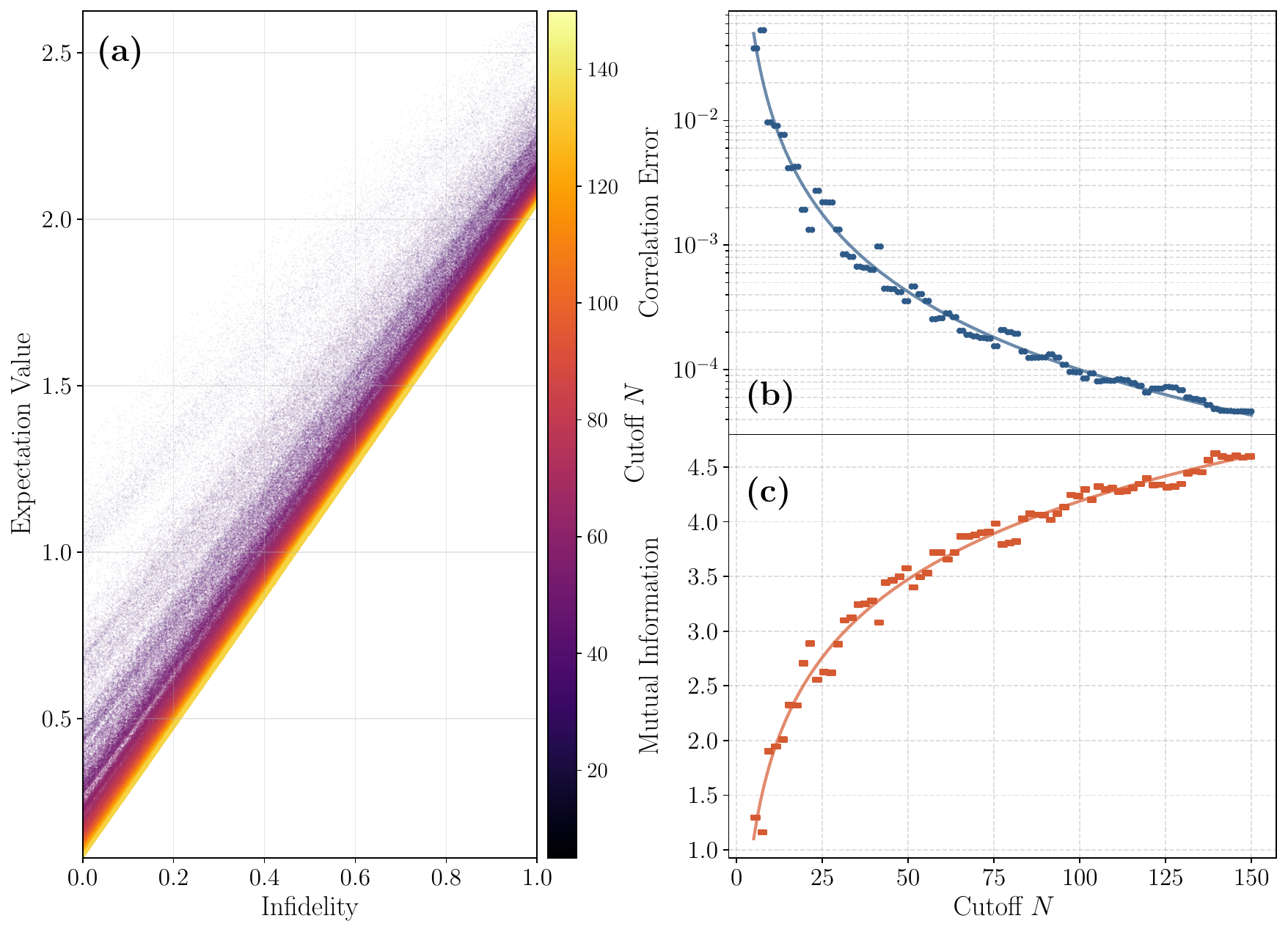}
	\caption{Visualization of the linear relationship between expectation values of \(\hat{O}_\mathrm{GKP}^{[N]}\) and infidelities for logical GKP qubits. \\\textbf{(a)} Scatter plot of expectation value and infidelity pairs for individual sampled points and \(N \in \{5, 6, \dots, 149, 150\}\) indicating a convergence to a linear relationship as \(N\rightarrow\infty\). \\\textbf{(b)}, \textbf{(c)} Scatter plots of the correlation error and mutual information between the expectation value and infidelity as a function of \(N\), respectively. Correlation error is defined as \(1-r\), where \(r\) is the Pearson correlation coefficient, and the mutual information is expressed in nats.}
	\label{fig:exp_val_corr}
\end{figure}

Analytically, this relationship becomes clear once we consider that in the logical GKP subspace, \(\langle \hat{O}_1\rangle\) vanishes and the remaining components act as logical Pauli operators. For a chosen \(\hat{O}_\mathrm{GKP}\) with \(\mathbf{u}_i=(u_{x,i},u_{y,i},u_{z,i})\), this implies (within the logical subspace)
\begin{equation}
	\langle\hat{O}_\mathrm{GKP}\rangle = \langle\hat{1}\rangle - \left( u_{x,i} \langle\hat{X}\rangle + u_{y,i} \langle\hat{Y}\rangle + u_{z,i} \langle\hat{Z}\rangle \right).
\end{equation}
Defining the Bloch vector of the evaluated probe state as \(\mathbf{u}_j = (u_{x,j},u_{y,j},u_{z,j}) = (\langle\hat{X}\rangle,\langle\hat{Y}\rangle,\langle\hat{Z}\rangle)\), we can rewrite this as
\(
\langle\hat{O}_\mathrm{GKP}\rangle = 1 - \mathbf{u}_i\cdot\mathbf{u}_j,
\)
which implies
\begin{equation}
	\langle\hat{O}_\mathrm{GKP}\rangle = 2\left(1-\mathcal{F}_{ij}\right),
\end{equation}
elegantly connecting the expectation value of the operator in the logical GKP subspace to the logical infidelity and confirming the numerical convergence, indicating that ground states of \(\hat{O}_\mathrm{GKP}^{[N]}\) for all possible \(\mathbf{u}\) define approximate logical GKP subspaces that converge to the ideal logical GKP subspace as \(N\rightarrow \infty\).

Consistent with the positive semidefiniteness established in Theorem~\ref{thm}, Fig.~\ref{fig:exp_val_corr} shows that the lowest expectation values of the truncated \(\hat{O}_\mathrm{GKP}^{[N]}\) remain positive and decrease monotonically toward zero as \(N\rightarrow\infty\), so the ideal state associated with \(\mathbf{u}\) is recovered as the unique zero-eigenvalue ground state in the limit.

In earlier GKP literature \cite{duivenvoordenSinglemodeDisplacementSensor2017, terhalScalableBosonicQuantum2020} one often uses effective squeezing parameters \(\Delta_x,\Delta_p\) to quantify the sharpness of the grid peaks in each quadrature. These are defined via the stabilizer expectations
\begin{equation}
	\Delta_x^2=\frac{1}{\pi}\ln\frac{1}{|\langle \hat{Z}^2\rangle|}, \qquad \Delta_p^2=\frac{1}{\pi}\ln\frac{1}{|\langle \hat{X}^2\rangle|}.
\end{equation}
For basis states, or more generally for any state obtained from them by a Gaussian transformation, the nonlinear squeezing \(\langle\hat{O}_\mathrm{GKP}\rangle\) reduces to the scaled sum of two effective squeezing parameters \cite{marekGroundStateNature2024}. For generic superpositions, including magic states, there is no analogous decomposition, making nonlinear squeezing a strict generalization of effective squeezing to non-stabilizer states.

Additionally, unlike individual effective squeezing parameters, nonlinear squeezing acts as a witness of non-Gaussianity. Letting \(\mathcal{G}\) denote the set of all Gaussian states, we have
\begin{equation}
	\min\limits_{\hat{\rho}\,\in\,\mathcal{G}}\trace\left[\hat{\rho}\,\hat{O}_\mathrm{GKP}(\mathbf{u})\right] = \frac{5}{3} - \norm{\mathbf{u}}_\infty,
\end{equation}
where \(\norm{\mathbf{u}}_\infty = \max\{|u_x|,|u_y|,|u_z|\}\leq 1\) is the infinity norm of \(\mathbf{u}\); see Supplemental Material~\cite{supplement} for the derivation. Similar bounds for different stellar ranks can be found numerically~\cite{provaznikWitnessesNonGaussianFeatures2026}.
\par

We have presented an efficient methodology for characterizing arbitrary qubits encoded in the GKP basis, bypassing the need for quantum state tomography. By generalizing the nonlinear squeezing paradigm, we constructed a parameterized family of positive semidefinite Hermitian operators $\hat{O}_\mathrm{GKP}(\mathbf{u})$, one for each point on the logical Bloch sphere, each with the ideal logical GKP superposition as its unique ground state. The construction unifies and generalizes existing results for the logical basis states~\cite{marekGroundStateNature2024} to the full Bloch sphere, including non-Clifford magic states such as $\ket{H_\mathrm{L}}$.

We established that the expectation value of $\hat{O}_\mathrm{GKP}$ provides a direct, measurable proxy for logical infidelity, converging to $2(1-\mathcal{F})$ across the entire logical Bloch sphere in the asymptotic limit. We further demonstrated that the finite-dimensional truncations $\hat{O}_\mathrm{GKP}^{[N]}$ offer a pathway for generating and evaluating physical approximations of arbitrary logical qubit states, providing targets for resource state preparation in finite-dimensional spaces. In addition, $\hat{O}_\mathrm{GKP}$ serves as a non-Gaussianity witness with a lower bound on the set of Gaussian states.

Because the target operators are built from the GKP stabilizers, estimating their expectation values requires only homodyne measurements of three quadratures, \(\hat{x}\), \(\hat{p}\) and \((\hat{x}-\hat{p})/\sqrt{2}\), rather than costly tomography. This makes the proposed framework both a mathematically rigorous metric for state quality and a highly practical tool for analyzing experimental data, benchmarking and optimizing state preparation circuits, and advancing the implementation of fault-tolerant continuous-variable quantum computers.

\begin{acknowledgments}
	We acknowledge support from the Czech Science Foundation (project 25-17472S). PM acknowledges European Union’s HORIZON Research and Innovation Actions under Grant Agreement no. 101080173 (CLUSTEC) and a grant from the Programme Johannes Amos Comenius under the Ministry of Education, Youth and Sports of the Czech Republic reg. no. CZ.02.01.01/00/22\_008/0004649. VK acknowledges IGA-PrF-2026-005 and thanks Giulia Ferrini for constructive feedback. We acknowledge use of the computational cluster at the Department of Optics. Data for generating the figures is openly available \cite{kucharDataEfficientCharacterization2026}.
\end{acknowledgments}

\newpage
\bibliographystyle{apsrev4-2}
\bibliography{resources/bibliography}

\clearpage
\pagestyle{plain}
\onecolumngrid
\appendix
\renewcommand{\appendixname}{Supplement}
\renewcommand{\thefigure}{S\arabic{figure}}
\renewcommand{\thetable}{S\Roman{table}}
\renewcommand{\theHtable}{S\Roman{table}}
\renewcommand{\theHfigure}{fig.app.\arabic{figure}}

\setcounter{figure}{0}
\setcounter{table}{0}
\vfill
\begin{center}
	\vspace*{1em}
	\textbf{\Large Supplemental Material for "Efficient characterization of general Gottesman-Kitaev-Preskill qubits"}\\[1em]

	Vojtech Kucha\v{r} and Petr Marek\\
	\textit{Department of Optics, Palack\'y University, 17. listopadu 1192/12, 779 00 Olomouc, Czech Republic}
\end{center}

\section{Unwrapping the Bloch Sphere}\label{app:sphere}
\markboth{}{}
\markright{}{}
In the first two supplements, we present a breakdown of how we performed the numerical analysis of approximate logical GKP states in truncated Fock spaces.

The basis of the numerical work was evenly covering the logical Bloch sphere with points. This is done by selecting a minimum angular separation \( \delta \) and following the algorithm below:

\begin{algorithm}[H]
	\caption{Uniform sampling of the Bloch sphere}
	\label{alg:bloch_sphere}
	\begin{algorithmic}[1]
		\State \textbf{Input:} \( \delta \) (minimum angular separation constraint)
		\State \textbf{Output:} Set of points \( \mathcal{P} \) on the unit sphere
		\State \textbf{Initialize:} \( \mathcal{P} \gets \) set of 26 core logical states, see Tab.~\ref{tab:logical_states}
		\State \textbf{Calculate:} \( N_{\text{max}} \gets \lceil 16 / \delta^2 \rceil \) (oversampled pool size)
		\State \textbf{Generate:} \( \mathcal{C} \gets \) set of \( N_{\text{max}} \) points forming a Fibonacci spherical spiral lattice \(\theta_j = \arccos\left( 1-\frac{2j+0.5}{N_{\text{max}}} \right), \phi_j = \frac{2\pi j}{\left( 1+\sqrt{5} \right)/2}\)
		\State \textbf{Shuffle:} Randomly permute the sequence of candidate points in \( \mathcal{C} \) to remove any spiral bias
		\State \textbf{For each} candidate point \( c \in \mathcal{C} \) \textbf{do}:
		\State \quad \textbf{If} \(\sphericalangle(c,p) \geq \delta\) for all points \( p \in \mathcal{P} \) \textbf{then}:
		\State \quad \quad \textbf{Add:} \( \mathcal{P} \gets \mathcal{P} \cup \{c\} \)
		\State \textbf{Return:} \( \mathcal{P} \)
	\end{algorithmic}
\end{algorithm}

The generated points are then sorted from the south pole using a greedy traveling salesperson (TSP) heuristic, i.e.

\begin{algorithm}[H]
	\caption{Greedy Nearest-Neighbor TSP Ordering}
	\label{alg:tsp_ordering}
	\begin{algorithmic}[1]
		\State \textbf{Input:} Set of sampled points \( \mathcal{P} \) on the sphere, size \( M \)
		\State \textbf{Output:} Ordered sequence of points \( \mathcal{O} \) forming a path
		\State \textbf{Initialize:} Empty sequence \( \mathcal{O} \)
		\State \textbf{Initialize:} Set of visited points \( \mathcal{V} \gets \emptyset \)
		\State \textbf{Find Start:} Identify \( p_{\text{start}} \in \mathcal{P} \) where \( p_z = \min_{p \in \mathcal{P}}(p_z) \)
		\State \textbf{Set:} Current point \( c \gets p_{\text{start}} \)
		\State \textbf{Add:} \( \mathcal{O} \gets (\, c \,) \), \( \mathcal{V} \gets \{c\} \)
		\State \textbf{While} \( |\mathcal{V}| < M \) \textbf{do}:
		\State \quad \textbf{Find Nearest:} Select \( p_{\text{next}} \in \mathcal{P} \setminus \mathcal{V} \) that minimizes the Euclidean distance \( \| c - p_{\text{next}} \|_2 \)
		\State \quad \textbf{Add:} Append \( p_{\text{next}} \) to \( \mathcal{O} \)
		\State \quad \textbf{Mark:} \( \mathcal{V} \gets \mathcal{V} \cup \{p_{\text{next}}\} \)
		\State \quad \textbf{Update:} \( c \gets p_{\text{next}} \)
		\State \textbf{Return:} \( \mathcal{O} \)
	\end{algorithmic}
\end{algorithm}	A graphical visualization of both the sampling and ordering algorithms is shown in Fig.~\ref{fig:sampling_visual}, with the 26 fixed logical states systematically listed in Tab.~\ref{tab:logical_states}, with example Wigner functions plotted in Fig.~\ref{fig:wigner_functions}. The greedy ordering algorithm naturally 'runs itself into a corner', resulting in increasing jumps between the points at the end of the sequence, which can be seen both in the bottom right corner of the heatmaps in Fig.~\ref{fig:exp_val} and at the end of the sampling series in Fig.~\ref{fig:sampling_visual}.\vfill

\begin{figure}[p]
	\centering

	\parbox{\textwidth}{
		\centering
		\renewcommand{\arraystretch}{0.55}
		\scriptsize
		\begin{tabular}{@{} c c rrr @{\hskip 25pt} c c rrr @{\hskip 25pt} c c rrr @{}}
			\toprule
			\multicolumn{5}{c}{\footnotesize \textbf{Stabilizer States}} & \multicolumn{5}{c}{\footnotesize \textbf{\(H\)-type Magic States}} & \multicolumn{5}{c}{\footnotesize \textbf{\(T\)-type Magic States}}                                                                                                                                                                                                                                                                                                              \\
			\cmidrule(r){1-5} \cmidrule(lr){6-10} \cmidrule(l){11-15}
			\footnotesize Label                                          & \footnotesize                                                      & \footnotesize $u_x$                                                & \footnotesize $u_y$ & \footnotesize $u_z$ & \footnotesize Label & \footnotesize                        & \footnotesize $u_x$ & \footnotesize $u_y$ & \footnotesize $u_z$ & \footnotesize Label & \footnotesize                           & \footnotesize $u_x$ & \footnotesize $u_y$ & \footnotesize $u_z$ \\
			\midrule
			$|0_\mathrm{L}\rangle$                                       & \textcolor{mplblue}{$\blacklozenge$}                               & 0                                                                  & 0                   & 1                   & $H_{+x+y}$          & \textcolor{mplgreen}{$\blacksquare$} & $1/\sqrt{2}$        & $1/\sqrt{2}$        & 0                   & $T_{+++}$           & \textcolor{mplpurple}{$\blacktriangle$} & $1/\sqrt{3}$        & $1/\sqrt{3}$        & $1/\sqrt{3}$        \\
			$|1_\mathrm{L}\rangle$                                       & \textcolor{mplblue}{$\blacklozenge$}                               & 0                                                                  & 0                   & $-1$                & $H_{+x-y}$          & \textcolor{mplgreen}{$\blacksquare$} & $1/\sqrt{2}$        & $-1/\sqrt{2}$       & 0                   & $T_{++-}$           & \textcolor{mplpurple}{$\blacktriangle$} & $1/\sqrt{3}$        & $1/\sqrt{3}$        & $-1/\sqrt{3}$       \\
			$|+_\mathrm{L}\rangle$                                       & \textcolor{mplblue}{$\blacklozenge$}                               & 1                                                                  & 0                   & 0                   & $H_{+x+z}$          & \textcolor{mplgreen}{$\blacksquare$} & $1/\sqrt{2}$        & 0                   & $1/\sqrt{2}$        & $T_{+-+}$           & \textcolor{mplpurple}{$\blacktriangle$} & $1/\sqrt{3}$        & $-1/\sqrt{3}$       & $1/\sqrt{3}$        \\
			$|-_\mathrm{L}\rangle$                                       & \textcolor{mplblue}{$\blacklozenge$}                               & $-1$                                                               & 0                   & 0                   & $H_{+x-z}$          & \textcolor{mplgreen}{$\blacksquare$} & $1/\sqrt{2}$        & 0                   & $-1/\sqrt{2}$       & $T_{+--}$           & \textcolor{mplpurple}{$\blacktriangle$} & $1/\sqrt{3}$        & $-1/\sqrt{3}$       & $-1/\sqrt{3}$       \\
			$|i_\mathrm{L}\rangle$                                       & \textcolor{mplblue}{$\blacklozenge$}                               & 0                                                                  & 1                   & 0                   & $H_{-x+y}$          & \textcolor{mplgreen}{$\blacksquare$} & $-1/\sqrt{2}$       & $1/\sqrt{2}$        & 0                   & $T_{-++}$           & \textcolor{mplpurple}{$\blacktriangle$} & $-1/\sqrt{3}$       & $1/\sqrt{3}$        & $1/\sqrt{3}$        \\
			$|-i_\mathrm{L}\rangle$                                      & \textcolor{mplblue}{$\blacklozenge$}                               & 0                                                                  & $-1$                & 0                   & $H_{-x-y}$          & \textcolor{mplgreen}{$\blacksquare$} & $-1/\sqrt{2}$       & $-1/\sqrt{2}$       & 0                   & $T_{-+-}$           & \textcolor{mplpurple}{$\blacktriangle$} & $-1/\sqrt{3}$       & $1/\sqrt{3}$        & $-1/\sqrt{3}$       \\
			                                                             &                                                                    &                                                                    &                     &                     & $H_{-x+z}$          & \textcolor{mplgreen}{$\blacksquare$} & $-1/\sqrt{2}$       & 0                   & $1/\sqrt{2}$        & $T_{--+}$           & \textcolor{mplpurple}{$\blacktriangle$} & $-1/\sqrt{3}$       & $-1/\sqrt{3}$       & $1/\sqrt{3}$        \\
			                                                             &                                                                    &                                                                    &                     &                     & $H_{-x-z}$          & \textcolor{mplgreen}{$\blacksquare$} & $-1/\sqrt{2}$       & 0                   & $-1/\sqrt{2}$       & $T_{---}$           & \textcolor{mplpurple}{$\blacktriangle$} & $-1/\sqrt{3}$       & $-1/\sqrt{3}$       & $-1/\sqrt{3}$       \\
			                                                             &                                                                    &                                                                    &                     &                     & $H_{+y+z}$          & \textcolor{mplgreen}{$\blacksquare$} & 0                   & $1/\sqrt{2}$        & $1/\sqrt{2}$        &                     &                                         &                     &                     &                     \\
			                                                             &                                                                    &                                                                    &                     &                     & $H_{+y-z}$          & \textcolor{mplgreen}{$\blacksquare$} & 0                   & $1/\sqrt{2}$        & $-1/\sqrt{2}$       &                     &                                         &                     &                     &                     \\
			                                                             &                                                                    &                                                                    &                     &                     & $H_{-y+z}$          & \textcolor{mplgreen}{$\blacksquare$} & 0                   & $-1/\sqrt{2}$       & $1/\sqrt{2}$        &                     &                                         &                     &                     &                     \\
			                                                             &                                                                    &                                                                    &                     &                     & $H_{-y-z}$          & \textcolor{mplgreen}{$\blacksquare$} & 0                   & $-1/\sqrt{2}$       & $-1/\sqrt{2}$       &                     &                                         &                     &                     &                     \\
			\bottomrule
		\end{tabular}
		\vskip 3pt
		\refstepcounter{table}\label{tab:logical_states}
		\footnotesize
		\textsc{Table \thetable.} Coordinates and visual markers for the 26 core logical states on the Bloch sphere.
		\vskip 5pt
	}

	\vfill

	\parbox{\textwidth}{
		\centering
		\includegraphics[width=\textwidth,height=0.47\textheight,keepaspectratio]{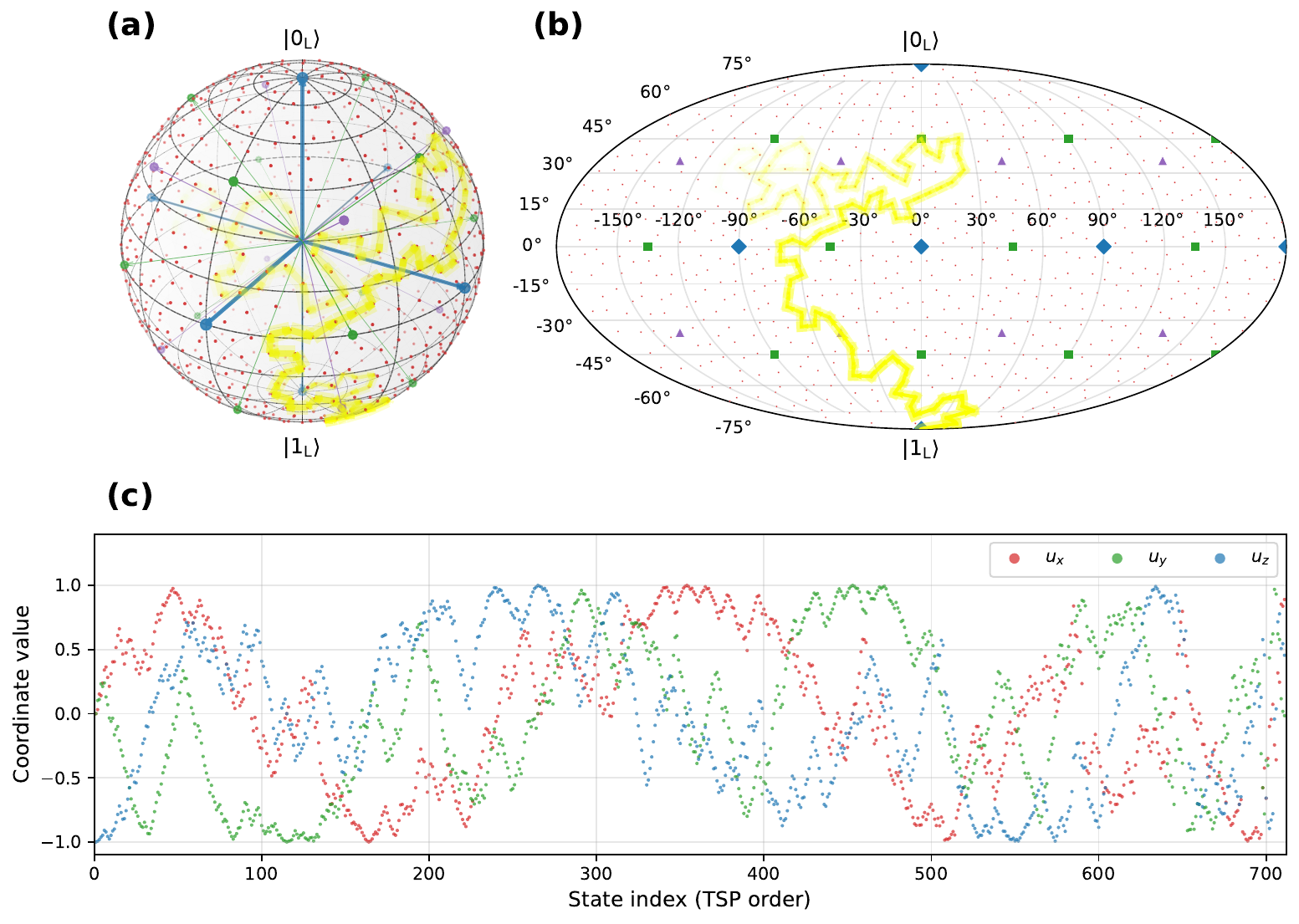}
		\caption{Graphical visualization of the sampling and ordering algorithms. \textbf{(a)} Orthographic 3D projection and \textbf{(b)} Mollweide projection of the Bloch sphere, evenly sampled with \( \delta = 0.1 \). The set of 26 core logical states is highlighted with {\color{mplblue}blue}, {\color{mplgreen}green} and {\color{mplpurple}purple} markers as per Tab.~\ref{tab:logical_states}, with the remaining Fibonacci sampled points in {\color{mplred}red}. First 100 steps of the greedy TSP ordering algorithm are highlighted using {\color{mplolive}yellow} arrows, the remaining steps can be seen coordinate-wise in \textbf{(c)}.}
		\label{fig:sampling_visual}
	}

	\vfill

	\parbox{\textwidth}{
		\centering
		\includegraphics[width=\textwidth,height=0.25\textheight,keepaspectratio]{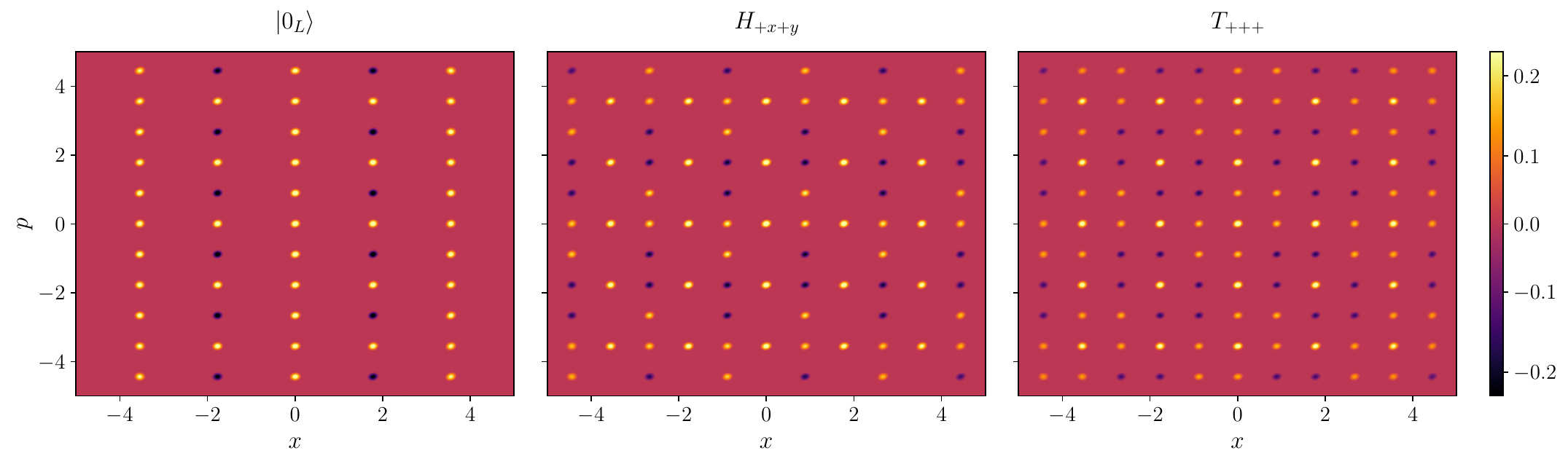}
		\caption{Example Wigner functions of core logical states, plotted as ground states of \(\hat{O}_\mathrm{GKP}^{[250]}(\mathbf{u})\).}
		\label{fig:wigner_functions}
	}
\end{figure}

\section{Converging to the Logical States}\label{app:convergence}
As outlined in the main text, by calculating the \(\hat{O}_\mathrm{GKP}^{[N]}\) operator and its ground state for each of the sampled points and subsequently evaluating the expectation value in all \(M^2\) state-operator combinations and comparing it with the logical infidelity for the respective coordinate pairs, we obtain data in Fig.~\ref{fig:exp_val_corr}.

To quantify the convergence indicated by Fig.~\ref{fig:exp_val_corr}a, we perform a linear regression analysis for each of the cutoff values \(N\). Denoting the expectation value as \(y\) and the infidelity as \(x\) for simplicity, we determine the correlation error as
\begin{equation}
	1 - r = 1 - \frac{\sum_{j=1}^{M^2} (x_j - \bar{x})(y_j - \bar{y})}{\sqrt{\sum_{j=1}^{M^2} (x_j - \bar{x})^2} \sqrt{\sum_{j=1}^{M^2} (y_j - \bar{y})^2}},
\end{equation}
where \(\bar{x}\) and \(\bar{y}\) represent the sample means of the evaluated points. We also calculate the mutual information in nats
\begin{equation}
	I(x; y) = \iint p(x,y) \ln\left( \frac{p(x,y)}{p(x)p(y)} \right) dx dy,
\end{equation}
where \(p(x,y)\) governs the theoretical joint continuous probability density function from which the expectation value and infidelity pairs were sampled. To avoid estimating this distribution, we utilize the Kraskov-Stögbauer-Grassberger algorithm, which calculates the mutual information based on the distances to the $k$-th nearest neighbors in the joint and marginal distributions~\cite{kraskovEstimatingMutualInformation2004}.

The linear regression analysis calculated for each individual cutoff dimension \(N\) provides a set of geometry-dependent slope parameters \(m(N)\) and vertical intercepts \(c(N)\), such that \(y \approx m(N)x + c(N)\). To extrapolate the ground-state convergence into the (\(N \to \infty\)) limit, we perform a secondary, asymptotic non-linear regression on these extracted coefficients.

As the cutoff \(N\) increases, the intercept physically decays toward zero, while the slope \(m(N)\) exhibits a smooth, monotonically saturating plateau. We model this slope convergence using a standard saturating power-law ansatz,
\begin{equation}
	m(N) = m_\infty - A N^{-d},
\end{equation}
where \(m_\infty\) corresponds to the true physical scaling limit, \(A\) is a positive amplitude, and \(d\) is the rate of saturation. To ensure that our extrapolated limit is robust against numerical artifacts and transient behaviors, we further performed a sensitivity meta-analysis by iteratively restricting the regression window to various subsets \(N \in [N_{\text{min}}, N_{\text{max}}]\) with \(N_\mathrm{min}>20\). By statistically averaging the extrapolated \(m_\infty\) values we obtained a limit of \(m_\infty = 2.00 \pm 0.02\), where the uncertainty reflects the one-standard-deviation variance of the values obtained for different windows.

\section{Deriving the Gaussian Bound}\label{app:witness}
In this section, we provide the detailed derivation of the Gaussian bound for the operator $\hat{O}_\mathrm{GKP}$ discussed in the main text.

Using the definitions of the operators from Eq.~\eqref{eq:gkppaulis}, we obtain the simple forms
\begin{subequations}
	\begin{align}
		 & \hat{O}_x = \cos(\hat{p}\sqrt{\pi}),                        \\
		 & \hat{O}_y = \cos\left( (\hat{x}-\hat{p})\sqrt{\pi} \right), \\
		 & \hat{O}_z = \cos(\hat{x}\sqrt{\pi}).
	\end{align}
\end{subequations}
Similarly, the penalty operator $\hat{O}_1$ from Eq.~\eqref{eq:Ooperators} can be expressed as
\begin{equation}
	\hat{O}_1 = \hat{1}-\frac{ \cos(2\sqrt{\pi}\hat{p}) + \cos\left( 2\sqrt{\pi}(\hat{x}-\hat{p}) \right) +\cos(2\sqrt{\pi}\hat{x})}{3}.
\end{equation}
Let us evaluate the expectation value of \(\hat{O}_\mathrm{GKP}\), defined in Eq.~\eqref{eq:general GKP qubit}, for a general state with the Wigner function \(W(x,p)\). Notice that
\begin{equation}
	\expval{\cos(\hat{x}\sqrt{\pi})} = \iint \d x\,\d p\;W(x,p)\cos(x\sqrt{\pi}) = \operatorname{Re}\left[\int \d^2 r\, W(r)e^{ix\sqrt{\pi}}\right],
\end{equation}
where we obtain the real part of the symmetrically ordered characteristic function \(\chi(\xi)\) at \(\xi=(\sqrt{\pi},0)\). Continuing in this manner with the remaining terms, we can write
\begin{equation}
	\expval{\hat{O}_\mathrm{GKP}} = 2 - R,
\end{equation}
where
\begin{equation}
	R(\chi) = \operatorname{Re}\left[\frac{1}{3}\left[ \chi(0,2\sqrt{\pi}) + \chi(2\sqrt{\pi},-2\sqrt{\pi}) + \chi(2\sqrt{\pi},0)\right] + u_x\chi(0,\sqrt{\pi}) + u_y \chi(\sqrt{\pi},-\sqrt{\pi}) + u_z\chi(\sqrt{\pi},0)  \right].
\end{equation}
Finding the Gaussian minimum is therefore identical to maximizing \(R\) for a Gaussian \(\chi\). The symmetrically ordered characteristic function of a general Gaussian with displacement \(\xi_0 = (x_0, p_0)\) and covariance matrix \(\Sigma\) with \(\Sigma + i\Omega/2 \succeq 0\) is
\begin{equation}
	\chi(\xi) = \exp\left( i\xi^T\xi_0 - \frac{1}{2}\xi^T\Sigma\xi \right) = \exp\left[i\left( x x_0 + p p_0 \right) - \frac{ x^2 \sigma_{xx} + 2 xp \sigma_{xp} + p^2 \sigma_{pp}}{2}\right].
\end{equation}
This implies
\begin{equation}
	\begin{aligned}
		R(\chi) = & \; \frac{1}{3}\left[ e^{-2\sigma_{xx}\pi}\cos(2\sqrt{\pi}x_0) + e^{-2\pi(\sigma_{xx}-2\sigma_{xp}+\sigma_{pp})}\cos\left( 2\sqrt{\pi}(x_0-p_0) \right) + e^{-2\sigma_{pp}\pi}\cos(2\sqrt{\pi}p_0) \right]            \\
		          & + u_z e^{-\frac{1}{2}\pi\sigma_{xx}}\cos(\sqrt{\pi}x_0) + u_y e^{-\frac{1}{2}\pi(\sigma_{xx}-2\sigma_{xp}+\sigma_{pp})}\cos\left( \sqrt{\pi}(x_0-p_0) \right)+u_x e^{-\frac{1}{2}\pi\sigma_{pp}}\cos(\sqrt{\pi}p_0).
	\end{aligned}
\end{equation}
For pure Gaussian states, we can additionally employ the parameterization
\begin{subequations}
	\begin{align}
		 & \sigma_{xx} = \frac{1}{2}\left( e^{-2r}\cos^2\theta + e^{2r}\sin^2\theta \right),                                                              \\
		 & \sigma_{xp} = \frac{1}{4}\left( e^{-2r}-e^{2r} \right)\sin 2\theta,                                                                            \\
		 & \sigma_{pp} = \frac{1}{2}\left( e^{-2r}\sin^2\theta + e^{2r}\cos^2\theta \right),                                                              \\
		 & \sigma_{xx} - 2\sigma_{xp} + \sigma_{pp} = \var (\hat{x}-\hat{p}) = \frac{1}{2}\left[ e^{-2r}(1-\sin 2\theta) + e^{2r}(1+\sin 2\theta)\right].
	\end{align}
\end{subequations}
Our goal is therefore to maximize \(R(\chi) = R(x_0,p_0,r,\theta)\) over \(u_x, u_y, u_z\) subject to \(u_x^2+u_y^2+u_z^2 = 1\). It is reasonable to expect that at least some optimal solutions will fully utilize Gaussian squeezing, so let us first examine the case of \(r = +\infty\). We obtain
\begin{subequations}
	\begin{align}
		 & R\xlongrightarrow[r\rightarrow +\infty,\;\theta = 0]{}\frac{1}{3}\cos(2\sqrt{\pi}x_0)+u_z\cos(\sqrt{\pi}x_0)                                                    \\
		 & R\xlongrightarrow[r\rightarrow +\infty,\;\theta = -\frac{\pi}{4}]{}\frac{1}{3}\cos\left(2\sqrt{\pi}(x_0 - p_0)\right)+u_y\cos\left(\sqrt{\pi}(x_0 - p_0)\right) \\
		 & R\xlongrightarrow[r\rightarrow +\infty,\;\theta = \frac{\pi}{2}]{}\frac{1}{3}\cos(2\sqrt{\pi}p_0)+u_x\cos(\sqrt{\pi}p_0).
	\end{align}
\end{subequations}
We can always choose appropriate \(x_0, p_0\) to maximize the cosine associated with the largest coefficient, therefore maximizing \(R\). As such, this infinitely squeezed solution gives us a first idea of the bound as
\begin{equation}
	R_{\mathrm{max},\infty} = \frac{1}{3} +\max\left( \abs{u_x},\abs{u_y},\abs{u_z} \right)\implies \expval{\hat{O}_\mathrm{GKP}}_{\mathrm{min},\infty} = \frac{5}{3}-\max\left( \abs{u_x},\abs{u_y},\abs{u_z} \right).
\end{equation}
Calculating \(\partial_r R\) allows for a direct proof that this is the Gaussian minimum in cases where only one coefficient is non-zero. For general sets of \(u_x, u_y, u_z\), such analytical proof is difficult to perform, however, numerical optimization over \(x_0,p_0,r,\theta\) for various coefficient sets strongly suggests that
\begin{equation}
	\min\limits_{\ket{\psi}\,\in\,\mathcal{G}_\mathrm{pure}}\expval{\hat{O}_\mathrm{GKP}(\mathbf{u})}{\psi} = \frac{5}{3}-\norm{\mathbf{u}}_\infty,
\end{equation}
where \(\mathbf{u} = (u_x, u_y, u_z)\) and where \(\mathcal{G}_\mathrm{pure}\) denotes the set of all pure Gaussian states. Since any mixed Gaussian state is a convex mixture of pure Gaussian states and the expectation value is linear, the bound extends naturally to the full set \(\mathcal{G}\) of all Gaussian states.

\AtEndDocument{\clearpage}
\end{document}